\documentclass{article}

\usepackage{graphicx,latexsym,amsfonts,amsmath,amsthm,xcolor,wrapfig,microtype}

\newcommand{\fig}[1]{\figurename~\ref{#1}}

\usepackage{bm}

\newif\ifComments
\Commentsfalse
\ifComments
\definecolor{purple}{hsb}{.77,1,1}
\newcommand{\marrow}{\marginpar[\hfill$\longrightarrow$]{$\longleftarrow$}}
\newcommand{\niceremark}[4]{\textcolor{#4}{\textsc{#1 #2:} \marrow\textsf{#3}}}
	\newcommand{\mati}[2][says]{\niceremark{Matias}{#1}{#2}{red}}
	\newcommand{\tom}[2][says]{\niceremark{Thomas}{#1}{#2}{magenta}}
	\newcommand{\birgit}[2][says]{\niceremark{Birgit}{#1}{#2}{purple}}
	\newcommand{\alex}[2][says]{\niceremark{Alex}{#1}{#2}{cyan}}
	\newcommand{\andre}[2][says]{\niceremark{Andr\'e}{#1}{#2}{blue}}
	\newcommand{\marcel}[2][says]{\niceremark{Marcel}{#1}{#2}{orange}}
	\newcommand{\highlight}[1]{\textcolor{purple}{#1}}
\else
	\newcommand{\mati}[2][says]{}
	\newcommand{\tom}[2][says]{}
	\newcommand{\birgit}[2][says]{}
	\newcommand{\alex}[2][says]{}
	\newcommand{\andre}[2][says]{}
	\newcommand{\marcel}[2][says]{}
	\newcommand{\highlight}[1]{}
\fi

\graphicspath{{figures/}}

\newtheorem{theorem}{Theorem}
\newtheorem{corollary}[theorem]{Corollary}
\newtheorem{lemma}{Lemma}
\newtheorem{construction}{{\bf Construction}} 

\newcommand{\floor}[1]{\ensuremath{\left \lfloor #1 \right \rfloor}}

\newcommand{\BE}{\ensuremath{{\sf BE}}}
\newcommand{\MST}{\ensuremath{{\sf MST}}}

\title{Packing Plane Spanning Graphs with Short Edges in Complete Geometric Graphs%
\footnote{%
A preliminary version of this paper was presented in the 27th International Symposium on Algorithms and Computation (ISAAC 2016), Sydney, Australia~\cite{DBLP:conf/isaac/AichholzerHKPRR16}.
}
}

\author{Oswin Aichholzer\thanks{Institute of Software Technology, Graz University of Technology, Austria. \texttt{\{oaich|thackl|apilz|bvogt\}@ist.tugraz.at}}
\and Thomas Hackl\footnotemark[2]
\and Matias Korman\thanks{Department of Computer Science, Tufts University, Medford, MA, USA. \texttt{matias.korman@tufts.edu}}
\and Alexander Pilz\footnotemark[2]
\and G\"unter Rote\thanks{Institut f\"ur Informatik, Freie Universit\"at Berlin, Germany. \texttt{rote@inf.fu-berlin.de}}
\and Andr\'e van Renssen\thanks{University of Sydney, Sydney, Australia. \texttt{andre.vanrenssen@sydney.edu.au}}
\and Marcel Roeloffzen\thanks{Department of Mathematics and Computer Science, TU Eindhoven, The Netherlands. {\tt m.j.m.roeloffzen@tue.nl}}
\and Birgit Vogtenhuber\footnotemark[2]
}

\begin{document}
\maketitle

\begin{abstract}
Given a set of points in the plane, we want to establish a connected spanning graph between these points, called connection network, that consists of several disjoint layers.
Motivated by sensor networks, our goal is that each layer is connected, spanning, and plane.
No edge in this connection network is too long in comparison to the length needed to obtain a spanning tree. 

We consider two different approaches.
First we show an almost optimal centralized approach to extract two layers.
Then we consider a distributed model in which each point can compute its adjacencies using only information about vertices at most a predefined distance away.
We show a constant factor approximation with respect to the length of the longest edge in the graphs.
In both cases the obtained layers are plane.
\end{abstract}

\section{Introduction}

Given a set $S$ of $n$ points in the plane and an integer $k$, we are interested in finding $k$ edge-disjoint non-crossing spanning graphs $H_1,H_2,\ldots, H_k$ on $S$ such that the length $\BE(H_1\cup H_2 \cup \cdots \cup H_k)$ of the
\emph{bottleneck edge} (the longest edge which is used) is as short as possible. 
Each subgraph $H_i$ is referred to as a \emph{layer} of the complete graph $G$ on the $n$ points.  We require each layer
to be non-crossing, but edges from different layers are allowed
to cross each other.  For $k=1$, the minimum spanning tree $\MST(S)$
solves the problem: its longest edge $\BE(\MST(S))$ is a lower
bound on the bottleneck edge of any spanning subgraph, and it is
non-crossing.  For larger~$k$, we take $\BE(\MST(S))$ as
the yardstick and measure the solution quality in terms of $\BE(\MST(S))$
and~$k$.

Although we find the problem to be of its own (theoretical) interest, this particular variation comes motivated from the field of sensor networks. In sensor networks, the energy consumed in transmission drastically grows as the distance between the two points increases~\cite{fwz-iar-05,k-mianbcr-12}. Since we want to avoid high energy consumption, it is desirable to apply the bottleneck criterion in order to minimize the maximum length of the whole network.

Once the network is built, we want to send messages through it without having to store the network explicitly at each node.  
One of the most commonly used methods for doing so is called {\em face routing}~\cite{ksu-crgn-99}, which uses only local information and guarantees delivery as long as the underlying network is plane. In fact, most local routing algorithms can only route on plane graphs. Extending these algorithms for non-plane graphs is a long-standing open problem in the field. In this paper, we provide a different way to avoid this obstacle. Rather than one plane graph, we construct several disjoint plane spanning graphs. If we split all the messages among the different layers we can potentially spread the load among a larger number of edges.


\subparagraph{Previous Work.}
This problem falls into the family of \emph{graph packing} problems,
where we are given a graph $G=(V,E)$ and a family $\mathcal{F}$ of
subgraphs of $G$. 
The aim is to pack as many
pairwise edge-disjoint subgraphs $H_1=(V,E_1), H_2=(V,E_2), \ldots $ as
possible into $G$. 

A related problem is the {\em decomposition} of $G$. In this case, we also look for disjoint subgraphs but require that $\bigcup_i E_i=E$. For example, there are known characterizations of when we can decompose the complete graph of $n$ points into paths~\cite{tarsi-dcmsp-83} (for $n$ even) and stars~\cite{pt-mdsm-05} (for $n$ odd). Dor and Tarsi~\cite{dt-gdnpc-97} showed that to determine whether we can decompose a graph $G$ into subgraphs isomorphic to a given graph $H$ is NP-complete. 
Concerning graph packing, Aichholzer {\em et al.}~\cite{ahkk-ppst-14} showed that we can pack $\Omega(\sqrt{n})$ edge-disjoint plane spanning trees in the complete graph on any set $S$ of $n$ points.
This bound has been improved to~$\lfloor n/3 \rfloor$ by Biniaz and Garc\'ia~\cite{biniaz_garcia}. 
Note that a trivial $\lfloor n/2 \rfloor$ upper bound follows from the number of edges in the complete graph.
Thus, the latter result is close to optimal.

In our case, the graph $G$ is the complete graph on a given point set $S$, and
$\mathcal{F}$ consists of all plane spanning graphs of $G$.
In addition to proving results for a large (fixed) number of layers, we are interested in minimizing a geometric constraint (Euclidean length of the longest edge among the selected graphs of~$\mathcal{F}$). To the best of our knowledge, this is the first packing problem of such type.  

\subparagraph{Results.}
Recall that both the point set $S$ and the integer $k$ are given and that we aim to find $k$ edge-disjoint connected plane spanning graphs $H_1,H_2,\ldots, H_k$ on $S$ such that the length $\BE(H_1\cup H_2 \cup \cdots \cup H_k)$ of the \emph{bottleneck edge} (the longest edge that is used) is minimized. 

We give two different approaches to solve the problem. In
Section~\ref{centralized} we give a construction for two spanning trees, i.e., $k=2$.
This construction is centralized in a classic model that assumes that the positions of all points are known and computed in a single place.
Our construction creates two trees and guarantees that all edges (except possibly one) have length at most $2\BE(\MST(S))$. The remaining edge has length at most $3\BE(\MST(S))$. We complement this construction with a matching worst-case lower bound that shows that for two spanning trees this is the shortest length the longest edge in the graphs can have.

In Section~\ref{sec_dist} we use a different approach to construct $k$ edge-disjoint connected plane spanning graphs (not necessarily trees). The construction works for any $k\leq n/12$ in an almost local fashion, i.e., using only information about vertices at most a certain maximum distance away. The only global information that is needed is $\beta$: $\BE(\MST(S))$ or some upper bound on it. Each point of $S$ can compute its adjacencies by only looking at nearby points, namely, those at distance $O(k\beta)$. 

A simple adversary argument shows that it is impossible to construct
spanning networks locally without knowing $\BE(\MST(S))$ (or an upper bound). The
lower bound of Section~\ref{centralized} shows that a
neighborhood of radius $\Omega(k\BE(\MST(S)))$ may be needed for the network, so we conclude that our
construction is asymptotically optimal in terms of the neighborhood.

For simplicity, throughout the paper we make the usual general position assumption that no three points are collinear. Without this assumption, it might be impossible to obtain more than a single plane layer (for example, when all points lie on a line).
Note however, that if collinear and partially overlapping edges are considered as non-crossing, our algorithms do not require the point set to be in general position.

\section{Centralized Construction} \label{centralized}

In this section we present a centralized algorithm to construct two layers.  We start with some properties on the minimum spanning tree of a set of points.

\begin{lemma}
  \label{longest-edge}
  If $|uw|>\max\{|uv|,|vw|\}$ for three points $u,v,w\in S$, the edge $uw$
  does not belong to $\MST(S)$.
\end{lemma}
\begin{proof}
  This is a special case of the more general well-known statement that
  the longest edges of any cycle in a graph, if it is unique, does not
  belong to its minimum spanning tree: The greedy algorithm would first pick all other
  edges of the cycle unless their endpoints are already
  connected. Thus, when the algorithm looks at the longest edge, its
  endpoints are already connected, and the edge is not included in the
  minimum spanning tree.
\end{proof}

\begin{lemma}\label{lem_not_in_triangle}
Let $S$ be a finite set of points in the plane and let $uv$ and $vw$ be two edges of $\MST(S)$.
Then the triangle $uvw$ does not contain any other point of $S$.
\end{lemma}


\begin{proof}
 Suppose for the sake of contradiction that the triangle
 $uvw$ contains a point~$p \in S$. Then the sum of the
 angles $vpu$
 and $vpw$ is at least $\pi$; see \fig{inside-wedge}~(a).
Hence, one of these angles, say, $vpu$ is least $\pi/2$. But then
$vu$ is the longest edge in the triangle $vpu$, and by
Lemma~\ref{longest-edge},
$vu$ cannot belong to $\MST(S)$, a contradiction.
\end{proof}

\begin{figure}[h]\label{inside-wedge}
  \centering
\includegraphics{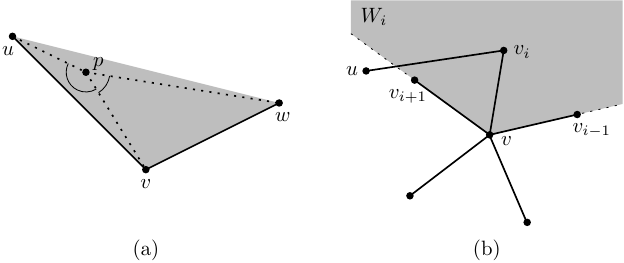}
\caption{(a)
  Proof of Lemma~\ref{lem_not_in_triangle}.
  (b)~Proof of Lemma~\ref{lem:stay-in-wedge}: The neighbors of
     $v_i$ cannot lie outside the wedge $W_i$ defined by its two siblings in $\MST(S)$.}
\end{figure}

\begin{lemma}\label{lem:stay-in-wedge}
Let $S$ be a finite set of points in the plane.
Let $v\in S$ be a point
with $k\geq3$ neighbors
 $v_0,\ldots,v_{k-1}$ in $\MST(S)$
 in counterclockwise order. 
Then for every triple $(v_{i-1},v_{i},v_{i+1})$ \textup(indices modulo $k$\textup),
the neighbors of $v_i$ in $\MST(S)$ are inside the wedge $W_i$ that is
bounded by the rays $vv_{i-1}$ and $vv_{i+1}$ and contains the edge $vv_i$.
\end{lemma}
\begin{proof}
  Assume for the sake of contradiction that
 $v_i$ has a neighbor
 $u$ in   $\MST(S)$ that does not lie in $W_i$;
 see \fig{inside-wedge}~(b).
  Then the edge $v_iu$ intersects one of the boundary rays of $W_i$,
  say,
$vv_{i+1}$.
As $\MST(S)$ is plane, the edge $vv_{i+1}$ does not intersect the edge  $v_iu$.
Hence, the point $v_{i+1}$ lies in the triangle $vv_iu$.
As $vv_i$ and $v_i
u$ are in $\MST(S)$, this contradicts Lemma~\ref{lem_not_in_triangle}.
\end{proof}

We denote by $\MST^2(S)$ the square of $\MST(S)$, the graph connecting
all pairs of points of $S$ that are at distance at most 2 in $\MST(S)$.
We call the edges of $\MST(S)$ \emph{short} edges and all
remaining edges of $\MST^2(S)$ \emph{long} edges.
For every long edge $uw$, 
the points $u$ and $w$
have a unique common neighbor $v$ in $\MST(S)$, which we call
the \emph{witness} of $uw$.
We define the \emph{wedge} of $uw$ to be the area that is
bounded by the rays $vu$ and $vw$ and contains the segment~$uw$.

We now characterize edge crossings in $\MST^2(S)$;
see \fig{crossing-MST2}.

\begin{figure}[h]\label{crossing-MST2}
  \centering
\includegraphics{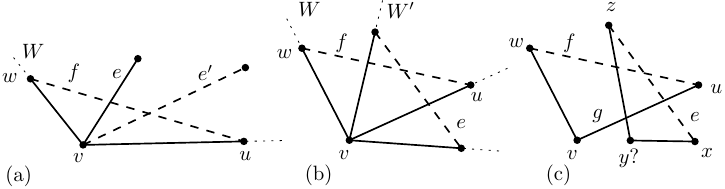}
\caption{The different cases of Lemma~\ref{lem:MST2edges}.
  Short
     edges are solid and long
     edges are dashed. (a)
     Condition 1, showing two options $e$ and $e'$, (b) Condition 2,
     and (c) the contradiction used in the proof.}
     \end{figure}

\begin{lemma}\label{lem:MST2edges}
Let $S$ be a finite set of points in the plane. 
Two edges $e$ and $f$ of $\MST^2(S)$ cross if and only if one of the following two conditions is fulfilled:
\begin{enumerate}
\item At least one of $e$ and $f$ is a long edge with witness $v$ and wedge $W$, and the other edge has $v$ as an endpoint and lies inside $W$.
\item\label{lem:MST2edges:cond_long} Both $e$ and $f$ are long edges with the same witness $v$, their wedges intersect, but none is contained in the other.
\end{enumerate} 
\end{lemma}
\begin{proof}
  Clearly, if both $e$ and $f$ are short,
they cannot cross.
Without loss of generality assume that
$f=uw$ is a long edge with witness $v$ and wedge $W$.
If $e$ is incident to $v$, then it must lie in $W$ in order
to cross $f$, and we satisfy Condition 1.

In the remaining case,
$e=xz$ with $x,z\in S\setminus\{u,v,w\}$.
By Lemma~\ref{lem_not_in_triangle}, $x$ and $z$ cannot lie in the
triangle
$uvw$; hence, $e$ must cross one of the $\MST$ edges $uv$ or $vw$ in
addition to the edge $f=uw$. It follows that $e$ cannot be short,
and it has some witness $y$ and some wedge $W'$.
We distinguish three possibilities for~$y$:

(i) If $y=v$,
we satisfy Condition 2:
$W'$ is not contained in $W$ because $e$ crosses
$uv$ or $vw$, and by swapping the roles of $e$ with $f$, we conclude
that $W$ is not contained in $W'$. The wedges $W$ and $W'$ must
overlap because otherwise $e$ and $f$ could not intersect.

(ii)
If $y=u$ or $y=w$, we can swap the roles of $e$ and $f$, thus satisfying Condition~1.

(iii)
We are left with the case that all six points $u,v,w,x,y,z$ are
distinct.
Let $g=uv$ or $g=vw$ be the edge that is intersected by $e$.
By  Lemma~\ref{lem_not_in_triangle}, the triangle
$xyz$ is empty; thus, $g$ must intersect a second edge $xy$ or $yz$ of
this triangle, in addition to $e=xz$. This is a contradiction, since
the edges $g$, $xy$, and $yz$ are edges of the $\MST$.

It is easy to see that the two conditions are
sufficient for a crossing:
In both situations of Condition 1 and Condition~2
(\fig{crossing-MST2}), if there were no crossing between $e$ and
 $f$, an endpoint of one edge would be contained in the triangle
 spanned by the other edge and its witness, contradicting
Lemma~\ref{lem_not_in_triangle}.
%
%
%
\end{proof}

\subsection{Constructing two almost disjoint layers}\label{s:almost-disjoint}
With the above observations we can proceed to show a construction that almost works for two layers: a single edge will be part of both layers, while all other edges occur in at most one tree.
To this end we consider the minimum spanning tree $\MST(S)$ to be rooted at an arbitrary leaf~$r$.
For any $v\in S$, we define its {\em level} $\ell(v)$ as its distance to $r$ in $\MST(S)$.
That is, $\ell(v)=0$ if and only if $v=r$.
Likewise, $\ell(v)=1$ if and only if $v$ is adjacent to $r$ etc.

For any $v\in S\setminus\{r\}$, we define its {\em parent} $p(v)$ as the first vertex traversed in the unique shortest path from $v$ to $r$ in $\MST(S)$.
Similarly, we define its {\em grandparent} $g(v)$ as $g(v)=p(p(v))$ if $\ell(v)\geq 2$ and as $g(v)=r$ otherwise (i.e., $g(v)=p(v)=r$ if $\ell(v) = 1$).
Each vertex $q$ for which $v=p(q)$ is called a \emph{child} of $v$.

\begin{figure}[h]
  \centering
\includegraphics{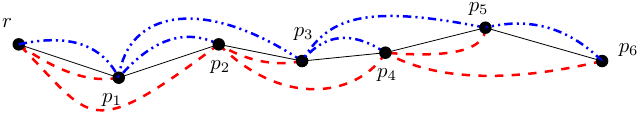}
   \caption{Example of Construction~\ref{con:dir-red-blue-doubletrees}. $\MST(S)$ is drawn in solid black, and the red and blue graphs in dashed and dot dashed, respectively. Note that the only common edge between the red and blue trees is the one from the root to its only neighbor in $\MST(S)$.}
   \label{fig:construc1}
\end{figure}

\begin{construction}\label{con:dir-red-blue-doubletrees}
Let $S$ be a finite set of points in the plane and let $\MST(S)$ be rooted at an arbitrary leaf $r\in S$.
We construct two graphs $R=G(S,E_R)$ and $B=G(S,E_B)$ as follows:
For any vertex $v_o\in S$ whose level is odd, we add the edge $v_op(v_o)$ to $E_R$ and the edge $v_og(v_o)$ to $E_B$.
For any vertex $v_e\in S\setminus\{r\}$ whose level is even, we add the edge $v_eg(v_e)$ to $E_R$ and the edge $v_ep(v_e)$ to $E_B$.
\end{construction}

For simplicity we say that the edges of $R=G(S,E_R)$ are colored red and the edges of $B=G(S,E_B)$ are colored blue.
An edge in both graphs is called red-blue.
See \fig{fig:construc1} for a sketch of the construction.

\begin{theorem}\label{thm:construction}
Let MST(S) be rooted at $r$. The two graphs $R=G(S,E_R)$ and $B=G(S,E_B)$ from Construction~\ref{con:dir-red-blue-doubletrees} fulfill the following properties:
\begin{enumerate}
\item\label{thm:construction.item1} Both $R$ and $B$ are plane spanning trees.
\item\label{thm:construction.item2} $\max\{\BE(R), \BE(B)\}\leq 2\BE(\MST(S))$.
\item\label{thm:construction.item3} $E_R\cap E_B=\{rs\}$, with $r=p(s)$, that is, $|E_R\cap E_B|=1$.
\end{enumerate}
\end{theorem}
\begin{proof}
Recall from Construction~\ref{con:dir-red-blue-doubletrees} that $r$ is a leaf of $\MST(S)$.
Hence $r$ has a unique neighbor $s$ in $\MST(S)$ and we have $r=p(s)=g(s)$ and $\ell(s)=1$.
Let \mbox{$S_o\!\subset\! S\!\setminus\!\{s\}$} be all \mbox{$v_o\!\in\! S$} whose level $\ell(v_o)$ is odd.
Likewise, let \mbox{$S_e\!\subset\! S\!\setminus\!\{r\}$} be all \mbox{$v_e\!\in\! S$} whose level $\ell(v_e)$ is even.
By construction, $E_R$ contains all the edges from odd-leveled nodes to their parents, those from even-leveled nodes to their grandparents and $rs$.
More formally,
\[
E_R = \bigcup_{v_o\in S_o}{\{v_op(v_o)\}} \cup \bigcup_{v_e\in S_e}{\{v_eg(v_e)\}} \cup \{rs\}.
\]
Similarly, $E_B$ contains edges from odd-leveled nodes to their grandparents, those from even-leveled nodes to their parents and $rs$, that is
\[
E_B = \bigcup_{v_o\in S_o}{\{v_og(v_o)\}}\cup \bigcup_{v_e\in S_e}{\{v_ep(v_e)\}} \cup \{rs\}.
\]
Thus, the edge $rs$ is the only shared edge between the sets $E_R$ and $E_B$, as stated in Property~\ref{thm:construction.item3} (we call this unique edge the {\em double-edge}).

As $E_R$ and $E_B$ are subsets of the edge set of~$\MST^2(S)$, the vertices of every edge in $E_R$ and $E_B$ have link distance at most~2 in~$\MST(S)$, and the bound on $\max\{\BE(R), \BE(B)\}$ stated in Property~\ref{thm:construction.item2} follows.

Further, both $R$ and $B$ are spanning trees, that is, connected and cycle-free graphs, as each vertex except~$r$ is connected either to its parent or grandparent in~$\MST(S)$.
To prove Property~\ref{thm:construction.item1}, it remains to show that both trees are plane.

Assume for the sake of contradiction that an edge $f$ is crossed by an edge~$e$ of the same color.
Recall that all edges of $E_R$ and $E_B$ are edges of~$\MST^2(S)$ whose endpoints have different levels.
By Lemma~\ref{lem:MST2edges}, at least one of $\{e,f\}$ has to be a long edge.
Without loss of generality let~$f = uw$ be a long edge and let $v$ be the witness of $f$ with $\ell(u) = \ell(v) - 1 = \ell(w) - 2$.
First note that $v$ cannot be an endpoint of $e$: since the level of $v$ has different parity than the one of $u$ and $w$, then $v$ must a leaf in this tree. Moreover, its only neighbor must be $u$ and thus the edges $uv$ and $f = uw$ cannot cross (and in particular $v$ cannot be an endpoint of $e$ as claimed). We further claim that $v$ cannot be the witness of $e$. Any edge that has $v$ as its witness is an edge from a child of $v$ to $u$ and therefore cannot cross $f = uw$.
As $e$ is neither incident to $v$ nor has $v$ as a witness, $e$ crossing $f$ contradicts Lemma~\ref{lem:MST2edges}.
This proves Property~\ref{thm:construction.item1} and concludes the
proof.
\end{proof}

The properties of our construction imply a first result stated in the following corollary.

\begin{corollary}\label{thm_almost}
For any set $S$ of $n$ points in the plane, there are two plane spanning trees $R=G(S,E_R)$ and $B=G(S,E_B)$ such that $|E_R\cap E_B|=1$ and $\max\{\BE(R), \BE(B)\}\leq 2\BE(\MST(S))$. 
\end{corollary}

Although the construction might seem to generalize to more layers by using edges of $\MST^k(S)$, this is not the case.
Already for $k=3$, we can show that the trees may be non-plane.
Take the example of \fig{fig:non-generalization}, where the full edges denote the minimum spanning tree.
However, if $a$ is chosen as the root of the tree, the edge $ad$ will be crossed by the edge from $e$ to either its parent, grandparent or great-grandparent.
In this example the problem can be remedied by choosing a different root.
But now consider placing a horizontally mirrored copy of this construction to the left so that $a$ and its mirrored version are connected by an edge.
Regardless of which root is chosen, in one of the two subtrees the node $a$ or its mirrored equivalent is the root of the respective subtree.
Hence, any root will create a crossing.

\begin{figure}
\centering
\includegraphics{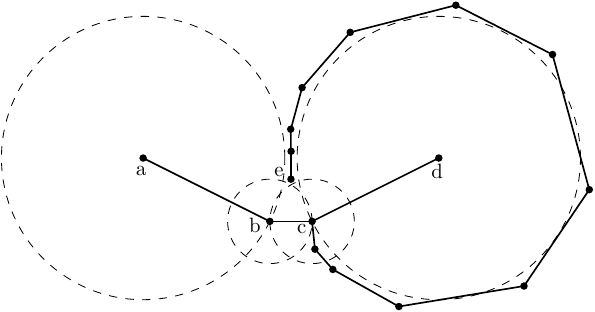}
\caption{Example graph where choosing $a$ as root creates a crossing when we generalize the above construction to three trees.}
\label{fig:non-generalization}
\end{figure}

\subsection{Avoiding the double edge}\label{s:avoid-edge}
Construction~\ref{con:dir-red-blue-doubletrees} is almost valid in the sense that only one edge was shared between both trees. In the following we modify this construction to avoid the shared edge. 

\begin{figure}
\centering
\includegraphics{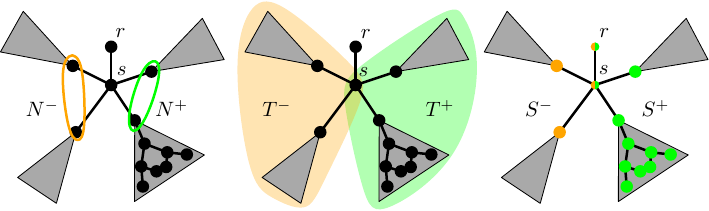}
\caption{Illustration of the various definitions used in Section~\ref{s:avoid-edge}. Grey triangles illustrate further subtrees, one shows interior vertices purely for illustrative purposes.}
\label{fig:plus-minus-trees}
\end{figure}

Let $N^-\subset (S\setminus\{r\})$ be the set of neighbors $v^-\in (S\setminus\{r\})$ of $s$ in $\MST(S)$ such that the ordered triangle $rsv^-$ is oriented clockwise.
Let $N^+\subset (S\setminus\{r\})$ be the set of neighbors $v^+\in (S\setminus\{r\})$ of $s$ in $\MST(S)$ such that the ordered triangle $rsv^+$ is oriented counterclockwise.
Let $T^-$ be the subtree of $\MST(S)$ that is connected to $s$ via the vertices in $N^-$ and let $T^+$ be the subtree of $\MST(S)$ that is connected to $s$ via the vertices in $N^+$.
Let $S^-\subset S$ consist of $r$ and the set of vertices from $T^-$ and let $S^+\subset S$ consist of $r$ and the set of vertices from $T^+$.
Observe that $S^-\cap S^+=\{r,s\}$ (see \fig{fig:plus-minus-trees}).
Let $E_R$ and $E_B$ be sets of red and blue edges as defined in the Construction~\ref{con:dir-red-blue-doubletrees}. Then let $E_R^-\subset E_R$ ($E_B^-\subset E_B$) be the subset of edges that have at least one endpoint in $S^-\setminus\{r,s\}$ and let $E_R^+\subset E_R$ ($E_B^+\subset E_B$) be the subset of edges that have at least one endpoint in $S^+\setminus\{r,s\}$.
Note that by this definition $E_R=E_R^-\cup E_R^+\cup\{rs\}$ and $E_B=E_B^-\cup E_B^+\cup\{rs\}$.
With this we define the subgraphs $R^-=G(S^-,E_R^-)$, $R^+=G(S^+,E_R^+)$, $B^-=G(S^-,E_B^-)$, and $B^+=G(S^+,E_B^+)$ and prove a useful non-crossing property between these graphs.

\begin{lemma}\label{lem:subgraphplanarity}
For any set $S$ of $n$ points in the plane, let $R=G(S,E_R)$ and $B=G(S,E_B)$ be the graphs from Construction~\ref{con:dir-red-blue-doubletrees}.
Then no edge in $E_R^-$ crosses an edge in $E_B^+$ and no edge in $E_R^+$ crosses an edge in $E_B^-$.
\end{lemma}
\begin{proof}
Consider any edge $e\in E_R^-$ that is not incident to $r$.
By Lemma~\ref{lem:MST2edges}, such an edge $e$ can be crossed only by an edge incident to at least one vertex of $S^-\setminus\{r,s\}$.
Hence, $e$ does not cross any edge of $E_B^+$.

Assume for the sake of contradiction that there is an edge $f\in E_B^+$ incident to $r$ that crosses an edge $e\in E_R^-$.
By construction, $e=rz$ is a long edge of $\MST^2(S)$ with witness $s$ and wedge $W$.
By Lemma~\ref{lem:MST2edges}, $f$ has to be incident to $s$, since $s$ cannot be the witness of any blue edges by construction. 
If $f$ is a short edge, then $f$ is not in $W$ by our definition of $S^-$ and $S^+$, which contradicts Lemma~\ref{lem:MST2edges}.
Hence, let $f=sc$ be a long edge of $\MST^2(S)$ with witness $b$. 
Following Lemma~\ref{lem:MST2edges}, the witness $b$ must be $s$, which contradicts the fact that $s$ cannot be a witness of any blue edge. 
This concludes the proof that no edge in $E_R^-$ is crossed by an edge in $E_B^+$.
Symmetric arguments prove that no edge in $E_R^+$ is crossed by an
edge in $E_B^-$.
\end{proof}

With this observation we can now prove that the two spanning trees (rooted at an arbitrary leaf $r$) from Construction~\ref{con:dir-red-blue-doubletrees} actually exist in 4 different color combination variants.

\begin{lemma}\label{lem:red-blue-tree-variants}
Let $S$ be a set of $n$ points in the plane.
Let $R=G(S,E_R)$ and $B=G(S,E_B)$ be the graphs from Construction~\ref{con:dir-red-blue-doubletrees} and let $R^-=G(S^-,E_R^-)$, $R^+=G(S^+,E_R^+)$, $B^-=G(S^-,E_B^-)$, and $B^+=G(S^+,E_B^+)$ be subgraphs as defined above.
Then $R$ and $B$ can be recolored to any of the four versions below, where each version fulfills the properties of Theorem~\ref{thm:construction}.
\begin{itemize}
\item[] (1) $R=G(S,E_R)$ and $B=G(S,E_B)$ (the \emph{original coloring})
\item[] (2) $R=G(S,E_B)$ and $B=G(S,E_R)$ (the \emph{inverted coloring})
\item[] (3) $R=G(S,E_B^-\cup E_R^+\cup\{rs\})$ and $B=G(S,E_R^-\cup E_B^+\cup\{rs\})$ (the \emph{$-$ side inverted coloring})
\item[] (4) $R=G(S,E_R^-\cup E_B^+\cup\{rs\})$ and $B=G(S,E_B^-\cup E_R^+\cup\{rs\})$ (the \emph{$+$ side inverted coloring})
\end{itemize}
\end{lemma}
\begin{proof}
The statement is trivially true for recolorings (1) and (2).
It is easy to observe that this corresponds to a simple recoloring.
Hence, Properties~\ref{thm:construction.item2}~and~\ref{thm:construction.item3} of Theorem~\ref{thm:construction} are also obviously true.
By Lemma~\ref{lem:subgraphplanarity}, both $R$ and $B$ are plane for the recolorings (3) and (4) and thus fulfill Property~\ref{thm:construction.item1} of Theorem~\ref{thm:construction} as well.
\end{proof}

With these tools we can now show how to construct two disjoint spanning trees. For technical reasons we use two different constructions based on the existence of a vertex $v$ in the minimum spanning tree where no two consecutive adjacent edges span an angle larger than~$\pi$.


\begin{theorem}\label{thm:nonpointed}
Let $S$ be a set of $n$ points in the plane, and $v$ a vertex of $S$. Assume that there is a minimum spanning tree $\MST(S)$ such that the angle between any two consecutive adjacent edges of $v$ in $\MST(S)$  is smaller than $\pi$.
Then there exist two plane spanning trees $R=G(S,E_R)$ and $B=G(S,E_B)$ such that $E_R\cap E_B=\emptyset$ and $\max\{\BE(R), \BE(B)\}\leq 2\BE(\MST(S))$.
\end{theorem}
\begin{proof}
We build the two spanning trees by using the vertex $v$ to decompose the minimum spanning tree into trees where $v$ is a leaf. For each of these subtrees we apply 
Construction~\ref{con:dir-red-blue-doubletrees} and possibly recolor them in one of the variants from Lemma~\ref{lem:red-blue-tree-variants}.

Let $S_v=\{v_1, \ldots, v_{k}\}$ be the set of vertices incident to $v$ in $\MST(S)$, labeled in counterclockwise order as they appear around $v$. 
Observe that $k\geq3$ is necessary to fulfill the angle condition from the theorem.
By Lemma~\ref{lem_not_in_triangle}, the convex hull of $S_v$ contains no points of $S$ except~$v$.
We start by constructing two plane spanning trees of $S_v \cup \{v\}$.  
The red spanning tree $R_v=G(S_v \cup \{v\},E_{vR})$ contains all edges incident to~$v$ except~$vv_1$, plus the edge $v_1v_2$, which lies on the convex hull boundary of $S_v$.
The blue spanning tree $B_v=G(S_v \cup \{v\},E_{vB})$ contains all edges on the convex hull boundary of $S_v$ except $v_1v_2$, plus the edge $vv_1$.
Observe that $R_v$ and $B_v$ are plane spanning trees, $E_{vR}\cap E_{vB}=\emptyset$, and $\max\{\BE(R_v), \BE(B_v)\}\leq 2\BE(\MST(S_v \cup \{v\}))$.

Next consider a vertex $v_i$ of $S_v$ and let $M_i$ be the maximal subtree of $\MST(S)$ that is connected to $v$ by $v_i$.
Let $S_i\subset S$ be the vertex set of~$M_i$.
Note that $M_i=\MST(S_i)$ and that~$v$ is a leaf in~$M_i$.
Thus, we can use Construction~\ref{con:dir-red-blue-doubletrees} to get spanning trees $R_i=G(S_i,E_{iR})$ and $B_i=G(S_i,E_{iB})$, all rooted at $v$.
The graphs $R_i$ and $B_i$ fulfill the three properties of Theorem~\ref{thm:construction} and 
the only edge shared between $R_i$ and $B_i$ is $vv_i$.

Considering Lemma~\ref{lem:MST2edges} and the fact that for $i\neq j$ the edges of $E_{iR}\cup E_{iB}$ have no point, except for the root $v$, in common with $E_{jR}\cup E_{jB}$, it is easy to see that no edge of $E_{iR}\cup E_{iB}$ crosses any edge of $E_{jR}\cup E_{jB}$.
In order to join the graphs to two plane spanning trees on~$S$, we adapt them slightly, while keeping the properties of Theorem~\ref{thm:construction}.
We first state how we combine the edge sets of the different plane spanning trees to get $R=G(S,E_R)$ and $B=G(S,E_B)$ and 
then reason why the claim in the theorem is true for this construction.
\[E_R = E_{vR} \cup (E_{3R}\!\setminus\!\{vv_3\}) \cup\ldots\cup
(E_{kR}\!\setminus\!\{vv_k\}) \cup (E_{1R}^-\cup E_{1B}^+) \cup
(E_{2B}^-\cup E_{2R}^+)\]
\[E_B = E_{vB} \cup (E_{3B}\!\setminus\!\{vv_3\}) \cup\ldots\cup
(E_{kB}\!\setminus\!\{vv_k\}) \cup (E_{1B}^-\cup E_{1R}^+) \cup
(E_{2R}^-\cup E_{2B}^+)\]

First we add the construction for $S_v \cup \{v\}$ to both edge sets.
This is the base to which all other trees will be attached.
Then the graphs from the subtrees $M_i$ for $1\leq i\leq k$ are added to this base.
The edges $vv_i$ are already used in $R_v$ or $B_v$, so we add the edges~$vv_i$ neither from $E_{iR}$ nor $E_{iB}$.
As both $v$ and $v_i$ are connected to both colors (both spanning trees), the construction stays connected.
As we did not add any additional edges the construction obviously stays cycle-free and the edge length bound is maintained.

It remains to argue the planarity of the resulting graphs.
By Lemma~\ref{lem:MST2edges}, edges of $E_{iR}$ or~$E_{iB}$ that cross any edge of $E_{vR}$ and $E_{vB}$ have to be incident to~$v$.
By Lemma~\ref{lem:stay-in-wedge}, only the edges $e_i^-=v_{i-1}v_i$ and $e_i^+=v_{i}v_{i+1}$ (indices modulo $k$) are crossed by edges of $E_{iR}\!\setminus\!\{vv_i\}$ and $E_{iB}\!\setminus\!\{vv_i\}$. 

Using the original coloring (see Lemma~\ref{lem:red-blue-tree-variants}) for $R_i$ and $B_i$ only red edges (edges of $E_{iR}\!\setminus\!\{vv_i\}$) cross $e_i^-$ and $e_i^+$.
For any $3\leq i\leq k$, $e_i^-$ and $e_i^+$ are blue, i.e., $e_i^-,e_i^+\in E_{vB}$.

For $i=1$, the edge $e_i^+$ is red.
In this case, we use the $+$ side inverted coloring (see Lemma~\ref{lem:red-blue-tree-variants}) for $R_i$ and $B_i$ (and exclude the edge $vv_i$):
$E_{iR}= E_{iR}^-\cup E_{iB}^+$ and $E_{iB}=E_{iB}^-\cup E_{iR}^+$.
Using this coloring, all shown properties remain valid (see Lemma~\ref{lem:red-blue-tree-variants}), all edges from~$R_i$ and~$B_i$ that cross the blue edge $e_i^-$ remain red, and all edges from $R_i$ and~$B_i$ that cross the red edge $e_i^+$ are now blue.

In a similar manner we fix the case of $i=2$, where the edge $e_i^-$ is red.
We use the $-$ side inverted coloring (see Lemma~\ref{lem:red-blue-tree-variants}) for $R_i$ and~$B_i$ (and exclude the edge $vv_i$):
$E_{iR}= E_{iB}^-\cup E_{iR}^+$ and $E_{iB}=E_{iR}^-\cup E_{iB}^+$.
Again, all shown properties remain valid (see Lemma~\ref{lem:red-blue-tree-variants}).
All edges from $R_i$ and $B_i$ that cross the red edge $e_i^-$ are now blue, and all edges from $R_i$ and~$B_i$ that cross the blue edge $e_i^+$ remain red.

Hence, with this slightly adapted construction (and coloring), $R=G(S,E_R)$ and $B=G(S,E_B)$ are plane spanning trees that solely use edges of $\MST^2(S)$ and have no edge in common.
\end{proof}

In the remaining case, for every vertex in an $\MST(S)$ there are two consecutive adjacent edges that span an angle larger than $\pi$.
In such an $\MST(S)$, every vertex has degree at most three,
since the angle between adjacent edges is at least $\pi/3$.

\begin{theorem}\label{thm:allpointed}
Consider a set $S$ of $n\geq 4$ points in the plane for which every vertex in the minimum spanning tree $\MST(S)$ has two consecutive adjacent edges spanning an angle larger than $\pi$. 
Then there exist two plane spanning trees $R=G(S,E_R)$ and
$B=G(S,E_B)$ such that $E_R\cap E_B=\emptyset$ and $\max\{\BE(R),
\BE(B)\}\leq 3\BE(\MST(S))$.
In addition, at most one edge of $E_R\cup E_B$ is longer than $2\BE(\MST(S))$.
\end{theorem}
\begin{proof}
As before in Theorem~\ref{thm:nonpointed} we will use our construction scheme for trees rooted at a leaf for the majority of the points and use a small local construction that avoids double edges. In this case, instead of removing a single vertex $v$ to decompose the tree we use a set of four vertices as follows.
We start at a leaf of $\MST(S)$ to generate a connected graph $P$ with four vertices that is a subgraph of $\MST(S)$.
Then we show how to construct $R$ and $B$ for $S$ for the different cases of $P$ in combination with the remainder of $\MST(S)$.

\textbf{The construction of $\bm{P=\{v_3,v_2,v_1,v_0\}}$:}
Let $v_3$ be a leaf of $\MST(S)$.
For the construction of $P$, we root $\MST(S)$ at $v_3$.
We call the number of vertices in the (sub)tree of which that vertex is a root of (including itself) the \emph{weight} of a vertex.
Hence, the weight of $v_3$ is $n$.
Further, the angle between two successive incident edges at a vertex of $\MST(S)$ that is larger than~$\pi$ is called the \emph{big angle}.

\begin{figure}[htb]
  \centering
  \includegraphics[page=1]{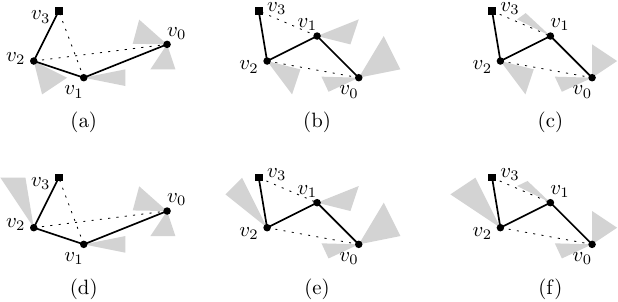}
  \caption{Case~\ref{thm:allpointed-item:1} for $P$ and the connections to the rest of $\MST(S)$.
	The gray triangles indicate possible subtrees of $\MST(S)$ and where and how they might be connected.
	Dotted edges are from $\MST^2(P)$.
	Note that the subtree with root $v_0$ can be on either side of the supporting line of $v_1v_0$ and even on both sides as indicated in the figure.
  }
  \label{fig:diffP1}
\end{figure}

Let $v_2$ be the unique child of $v_3$ in $\MST(S)$ (i.e., $v_3=p(v_2)$).
To define $v_1$ and $v_0$ we use a case distinction. 
	Consider the set $C$ of the children of $v_2$ that are not spanning the big angle with $v_3$ at $v_2$. (Note that $v_3$ may or may not be spanning the big angle at $v_2$ and that $C$ contains 0, 1, or 2 vertices). 
\begin{enumerate}
\item\label{thm:allpointed-item:1}
	If $v_2$ has only a single child (i.e., $C$ is empty), or if $C$ contains a vertex that is not a leaf in $\MST(S)$, we choose it (or one of them) as $v_1$.
We assume w.l.o.g.\ that $v_1$ is the successor of $v_3$ in clockwise order around $v_2$.
Further, we choose $v_0$ as a child of $v_1$ such that $v_2$ and $v_0$ are consecutive around $v_1$ and do not span the big angle at $v_1$.
If $v_1$ has two children, and this is true for both, then we choose $v_0$ such that it is the successor of $v_2$ in counterclockwise order around $v_1$.
See \fig{fig:diffP1}~(a)--(f) for the six different variations of this case, i.e., all possible distributions of the positions of the big angles. Explicit definitions for these cases can be found below.
W.l.o.g., we require the subtree at $v_2$ to be nonempty in cases (d)--(f) and the subtree of $v_1$ to be nonempty in cases (c) and (f).
For all cases we assume without loss of generality that the angle $v_3v_2v_1$ is clockwise less than $\pi$.
\begin{itemize}
\item[(1a)] The edge $v_3v_2$ is \textbf{adjacent} to the big angle at $v_2$ and the angle $v_2v_1v_0$ is clockwise \textbf{less} than $\pi$. 
\item[(1b)] The edge $v_3v_2$ is \textbf{adjacent} to the big angle at $v_2$, the angle $v_2v_1v_0$ is clockwise \textbf{greater} than $\pi$, and the edge $v_2v_1$ is \textbf{adjacent} to the big angle at $v_1$. 
\item[(1c)] The edge $v_3v_2$ is \textbf{adjacent} to the big angle at $v_2$, the angle $v_2v_1v_0$ is clockwise \textbf{greater} than $\pi$, and the edge $v_2v_1$ is \textbf{not adjacent} to the big angle at $v_1$. 
\item[(1d)] The edge $v_3v_2$ is \textbf{not adjacent} to the big angle at $v_2$ and the angle $v_2v_1v_0$ is clockwise \textbf{less} than $\pi$. 
\item[(1e)] The edge $v_3v_2$ is \textbf{not adjacent} to the big angle at $v_2$, the angle $v_2v_1v_0$ is clockwise \textbf{greater} than $\pi$, and the edge $v_2v_1$ is \textbf{adjacent} to the big angle at $v_1$. 
\item[(1f)] The edge $v_3v_2$ is \textbf{not adjacent} to the big angle at $v_2$, the angle $v_2v_1v_0$ is clockwise \textbf{greater} than $\pi$, and the edge $v_2v_1$ is \textbf{not adjacent} to the big angle at $v_1$. 
\end{itemize}

\item\label{thm:allpointed-item:2} $C$ contains at least one vertex and all vertices in $C$ are leaves in $\MST(S)$. 
Note that this implies that $v_2$ has degree exactly three in $MST(S)$, as the degree of any vertex with a big angle in a minimum spanning tree is at most three
and thus $C$ would be empty if $v_2$ had degree at most two.
We choose a vertex of $C$ as $v_1$ (assuming w.l.o.g.\ that $v_1$ is the successor of $v_3$ in clockwise order around $v_2$), and choose the other child of $v_2$ as $v_0$.  
Note that if $n \geq 5$ then $v_0$ cannot be a leaf in $\MST(S)$ and hence $v_0$ spans a big angle with $v_3$ at $v_2$.
Therefore, taking the location of the big angle at $v_2$ into account, there are two possibilities for the counterclockwise order of incident edges around $v_2$.
See \fig{fig:diffP2}~(a) and~(b).
\begin{itemize}
\item[(2a)] The angle $v_1v_2v_0$ is clockwise \textbf{less} than $\pi$ (the edge $v_3v_2$ is \textbf{adjacent} to the big angle at $v_2$.) 
\item[(2b)] The angle $v_1v_2v_0$ is clockwise \textbf{greater} than $\pi$ (the edge $v_3v_2$ is \textbf{not adjacent} to the big angle at $v_2$).
Here, both $v_1$ and $v_0$ must be leaves implying that $n=4$ and we do not have any subtrees. 
\end{itemize}	
\end{enumerate}

\begin{figure}[htb]
  \centering
  \includegraphics[page=2]{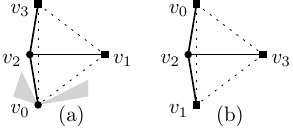}
  \caption{Case~\ref{thm:allpointed-item:2} for the selection of $P$ and the connections to the rest of $\MST(S)$.
    Vertices drawn with squares are leaves of $\MST(S)$.
    The gray triangles indicate possible subtrees of $\MST(S)$ and where they might be connected.
    Dotted edges are from $\MST^2(P)$.
    Note that the subtree with root $v_0$ can be on either side of the supporting line of $v_1v_0$ and even on both sides as indicated in the figure.
    Further note that case (b) can appear only if $S$ contains exactly $4$ vertices.
  }
  \label{fig:diffP2}
\end{figure}

\textbf{The construction of $\bm{R}$ and $\bm{B}$:}
First we show how to construct the trees $R_P=G(\{v_3,v_2,v_1,v_0\}, E_{PR})$ and $B_P=G(\{v_3,v_2,v_1,v_0\},E_{PB})$.
The vertices of $P$ can either be in convex position or form a triangle with one interior point, with $v_1$ interior for the cases shown in \fig{fig:diffP1}~(b) and~(e), and $v_2$ interior for the cases shown in \fig{fig:diffP1}~(e) and~(f);
there are no other non-convex versions: otherwise either the path $v_3,v_2,v_1,v_0$ could not be in $\MST(S)$,
or one of the vertices of $P$ could not be incident to a big angle (recall Lemma~\ref{lem_not_in_triangle}).
As in any non-convex case the complete graph on $\{v_3,v_2,v_1,v_0\}$ is crossing-free, 
any construction of $R_P$ and $B_P$ for the convex cases is also valid for the non-convex cases.
Let us thus assume the points of $P$ to be in convex position.
For Cases~1a~and~1d the vertices $\{v_0,v_1,v_2,v_3\}$ must appear in this order on their convex hull and we set $R_P = \{v_2v_3, v_3v_1, v_1v_0\}$ and $B_P = \{v_3v_0, v_0v_2, v_2v_1\}$, as illustrated in \fig{fig:colP}~(a).
For Cases~1b, 1c, 1e, and 1f the vertices appear in the order $\{v_3,v_1,v_0,v_2\}$ on their convex hull.
(The order $\{v_0,v_1,v_2,v_3\}$ violates the fact that the clockwise angle $v_0v_1v_2$ is not large and the order $\{v_0, v_3,v_1,v_2\}$ has edges $v_0v_1$ and $v_2v_3$ crossing, which are both edges of $\MST(S)$; no other orderings exist when accounting for symmetry.)
For these cases we set $R_P = \{v_3v_2, v_2v_1, v_1v_0\}$ and $B_P = \{v_1v_3, v_3v_0, v_0v_2\}$, as in \fig{fig:colP}~(b).
All edges except the edge $v_3v_0$ (in $E_{PB}$) are from $\MST^2(S)$ and have endpoints with different levels in $\MST(S)$ rooted at $v_3$. In contrast, $v_3v_0$ is an edge of $\MST^3(S)$, which could be crossed by other edges of the construction. We will later discuss how to handle this.

For Case~\ref{thm:allpointed-item:2} the placement must be convex as $v_0,v_1$ and $v_3$ are adjacent to $v_2$ and the clockwise $v_0v_2v_3$ and $v_1v_2v_0$ are convex for Cases 2a and 2b, respectively.
Since the ordering around $v_2$ is fixed (modulo symmetry) by the case definition, the vertices appear in the order $\{v_2,v_3,v_1,v_0\}$ for Case 2a and $\{v_2,v_0,v_3,v_1\}$ for Case 2b.
For Case 2a we set $R_P = \{v_2v_3,v_3v_0,v_0v_1\}$ and $B_P = \{v_3v_1,v_1v_2,v_2v_0\}$ and for Case 2b we set $R_P = \{v_2v_0,v_0v_1,v_1v_3\}$ and $B_P = \{v_1v_2,v_2v_3,v_3v_0\}$ as illustrated in \fig{fig:colP}~(c)~and~(d).
For both cases, all edges are from $\MST^2(S)$.

\begin{figure}[htb]
  \centering
  \includegraphics[page=3]{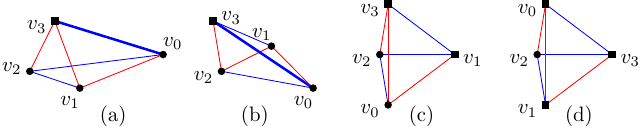}
  \caption{The different colorings for $P$: (a) for the cases from \fig{fig:diffP1}~(a) and~(d), (b) for the cases from \fig{fig:diffP1}~(b), (c), (e), and (f), (c) for the case from \fig{fig:diffP2}~(a), and (d) for the case from \fig{fig:diffP2}~(b).
	  The bold edge in (a) and (b) is an edge of $\MST^3(S)$, that is, an edge with link distance $3$ in $\MST(S)$.
  }
  \label{fig:colP}
\end{figure}

With $R_P$ and $B_P$ as a base, we now create red and blue trees for all remaining subtrees of $MST(S)$ and ``attach'' them to the base.
For Case~\ref{thm:allpointed-item:1} we define three possible subtrees.
Let $T_0'=G(S_0',E_0')$ be the subtree (i.e., connected component) of $\MST(S)$ that contains $v_0$ when removing $v_1$ (and its incident edges) from $\MST(S)$.
Likewise, let $T_1'=G(S_1',E_1')$ be the subtree of $\MST(S)$ that contains $v_1$ when removing $v_0$ and $v_2$ from $\MST(S)$, and
let $T_2'=G(S_2',E_2')$ be the subtree of $\MST(S)$ that contains $v_2$ when removing $v_1$ and $v_3$ from $\MST(S)$.
For Case~\ref{thm:allpointed-item:2}(a) there is one possible subtree $T_0'=G(S_0',E_0')$, which is the subtree of $\MST(S)$ that contains $v_0$ when removing $v_2$ from $\MST(S)$. (Case~\ref{thm:allpointed-item:2}(b) appears only if $n=4$ and hence the construction is already completed.)
The subtrees $T_0'$, $T_1'$, and $T_2'$ are shown as (pairs of) gray triangles in \figurename~\ref{fig:diffP1} and~\ref{fig:diffP2}.
To connect these subtrees to the bases~$R_P$ and~$B_P$, we create corresponding trees $T_0$, $T_1$, and $T_2$ depending on the different shown cases, then apply Construction~\ref{con:dir-red-blue-doubletrees} to them, and possibly recolor them using Lemma~\ref{lem:red-blue-tree-variants}.

We first consider the different subtrees for Case~\ref{thm:allpointed-item:1}. In essence, for each tree we pick a neighbor from $\{v_3,v_2,v_1,v_0\}$ to add to $T_0',T_1',T_2'$, which we then use as root for Construction~\ref{con:dir-red-blue-doubletrees}.
When there is a choice we pick a root that is adjacent to the outgoing edge from $v_i$ into the subtree $T_i'$ as defined more precisely below. 

\textbf{$\bm{T_0}$:} For all cases, we consider the subtree $T_0=G(S_0,E_0)$ of $\MST(S)$, with $S_0=S_0'\cup\{v_1\}$, $E_0=E_0'\cup\{v_1v_0\}$.
We root $T_0$ at $r = v_1$ (observe, $v_1$ is a leaf in $T_0$ with unique child $s=v_0$) and apply Construction~\ref{con:dir-red-blue-doubletrees} to get $R_0=G(S_0,E_{0R})$ and $B_0=G(S_0,E_{0B})$, with the double-edge $rs$ removed from both $E_{0R}$ and~$E_{0B}$.

\textbf{$\bm{T_1}$:} For the cases depicted in \fig{fig:diffP1}~(a), (b), (d), and~(e), we define $T_1=G(S_1,E_1)$ of $\MST(S)$, such that $S_1=S_1'\cup\{v_0\}$, $E_1=E_1'\cup\{v_1v_0\}$.
We root $T_1$ at $r = v_0$ (observe, $v_0$ is a leaf in $T_1$ with unique child $v_1$) and apply Construction~\ref{con:dir-red-blue-doubletrees} to get $R_1=G(S_1,E_{1R})$ and $B_1=G(S_1,E_{1B})$, with the double-edge $rs$ removed from both $E_{1R}$ and $E_{1B}$.

In the cases shown in \fig{fig:diffP1}~(c)~and~(f), we define $T_1=G(S_1,E_1)$ of $\MST(S)$, such that $S_1=S_1'\cup\{v_2\}$, $E_1=E_1'\cup\{v_1v_2\}$.
We root $T_1$ at $r = v_2$ (observe, $v_2$ is a leaf in $T_1$ with unique child $v_1$) and apply Construction~\ref{con:dir-red-blue-doubletrees} to get $R_1=G(S_1,E_{1R})$ and $B_1=G(S_1,E_{1B})$, with the double-edge $rs$ removed from both $E_{1R}$ and $E_{1B}$.

\textbf{$\bm{T_2}$:} For the cases depicted in \fig{fig:diffP1}~(a)--(c), let $T_2=G(S_2,E_2)$ be a subtree of $\MST(S)$, such that $S_2=S_2'\cup\{v_1\}$, $E_2=E_2'\cup\{v_1v_2\}$.
We root $T_2$ at $r = v_1$ (observe, $v_1$ is a leaf in $T_2$ with unique child $v_2$) and apply Construction~\ref{con:dir-red-blue-doubletrees} to get $R_2=G(S_2,E_{2R})$ and $B_2=G(S_2,E_{2B})$, with the double-edge $rs$ removed from both $E_{2R}$ and $E_{2B}$.

For the cases depicted in \fig{fig:diffP1}~(d)--(f), $T_2=G(S_2,E_2)$ is the subtree of $\MST(S)$, such that $S_2=S_2'\cup\{v_3\}$, $E_2=E_2'\cup\{v_3v_2\}$.
We root $T_2$ at $r = v_3$ (observe, $v_3$ is a leaf in $T_2$ with unique child $v_2$) and apply Construction~\ref{con:dir-red-blue-doubletrees} to get $R_2=G(S_2,E_{2R})$ and $B_2=G(S_2,E_{2B})$, with the double-edge $rs$ removed from both $E_{2R}$ and $E_{2B}$.

\medskip
It is easy to see that the edge sets $E_{PR}$, $E_{0R}$, $E_{1R}$, $E_{2R}$, $E_{PB}$, $E_{0B}$, $E_{1B}$, and $E_{2B}$ are all individually edge~disjoint.
In the following, we describe how these edge sets are combined to form the two plane spanning trees $R$ and $B$ in the different cases; see~\fig{fig:combinedRB1} for the convex versions and \fig{fig:combinedRB1nonconv} for the non-convex versions of the corresponding cases illustrated in \fig{fig:diffP1}.
For the non-convex cases only points $v_1$ and $v_2$ can be in the interior as $v_3$ or $v_0$ being in the interior would violate Lemma~\ref{lem_not_in_triangle}. 
Furthermore, in some of the cases (a)--(f), further restrictions apply as listed below.
\begin{itemize}
\item[(a)] Neither $v_1$ nor $v_2$ can be the middle point as both clockwise angles $v_3v_2v_1$ and $v_2v_1v_0$ have angle less than $\pi$.
\item[(b)] Only $v_1$ can be in the center (the subtree of $v_2$ is nonempty and hence $v_2$ would not be incident to a big angle). 
\item[(c)] Neither $v_1$ nor $v_2$ can be in the center (both subtrees are nonempty).
\item[(d)] Similar to (a).
\item[(e)] Both $v_1$ and $v_2$ may be the middle point.
\item[(f)] Only $v_2$ can be in the center (the subtree of $v_1$ is non-empty).
\end{itemize}

\begin{figure}[htb]
  \centering
  \includegraphics[page=4]{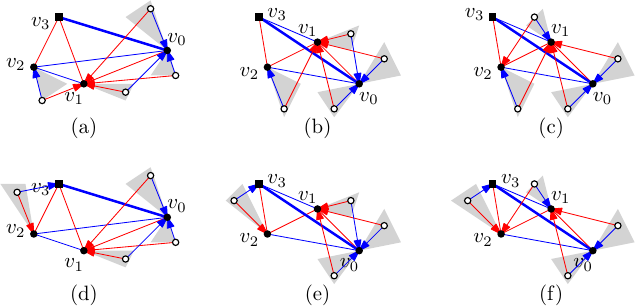}
  \caption{The different plane spanning trees $R$ and $B$ for case~\ref{thm:allpointed-item:1} when $P$ is in convex position.
    The bold blue edges $v_3v_0$ are edges of $\MST^3(S)$, i.e., edges with link distance $3$ in $\MST(S)$, that might still be crossed.
    Note that $v_3v_0$ can only be crossed in cases~(a) and~(b), as in the other cases this edge is surrounded by ``uncrossable'' $\MST^2(S)$-edges.
  }
  \label{fig:combinedRB1}
\end{figure}

\begin{figure}[htb]
  \centering
  \includegraphics[page=7]{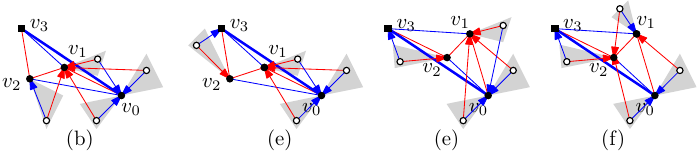}
  \caption{The different plane spanning trees $R$ and $B$ for case~\ref{thm:allpointed-item:1} when $P$ is not in convex position, with $v_1$ or $v_2$ 
	  in the interior. The case numberings are the same as the ones in \fig{fig:diffP1}.
  }
  \label{fig:combinedRB1nonconv}
\end{figure}

First we add the red and blue trees for $T_0$.
By construction, only edges of $E_{0R}$ connect to $v_1$ (only crossing edges of $E_{PB}$) and the edges of $E_{0B}$ do not cross any edge outside $T_0$.
For the cases (a), (b), (d), and (e)
we use the inverted coloring (see Lemma~\ref{lem:red-blue-tree-variants}) for the two trees of $T_1$.
For the remaining cases (c) and (f)
we use the original coloring (see Lemma~\ref{lem:red-blue-tree-variants}) for the two trees of $T_1$.
For adding the red and blue trees for $T_2$ we use the original coloring for the cases (a-c), and the inverted coloring for the cases (d-f). 
For the sake of simplicity, we exchange the names of the according sets $E_{iR}$ and $E_{iB}$ whenever we use the inverted coloring for a $T_i$, $i\in \{1,2\}$.

Since the edge sets $E_{0R}, E_{1R}$, and $E_{2R}$ form spanning trees connecting the nodes of $T_0$, $T_1$, and $T_2$ and since $E_{PR}$ connects their roots and $v_3$ with a spanning tree, it follows that $R=G(S,E_{PR}\cup E_{0R}\cup E_{1R}\cup E_{2R})$ is a spanning tree for $S$. The same argument applies to show that $B=G(S,E_{PB}\cup E_{0B}\cup E_{1B}\cup E_{2B})$ is a spanning tree. All edges in $E_{0R}, E_{1R}, E_{2R}, E_{0B},E_{1B}$, and $E_{2B}$ are from $\MST^2(S)$ and all edges of $E_{PR}$ and $E_{PB}$ are from $\MST^2(S) \cup \{v_3v_0\}$, it follows that $\max\{\BE(R), \BE(B)\}\leq 3\BE(\MST(S))$, with only the edge $v_3v_0$, which may occur in $E_{PB}$ being possibly larger than $2\BE(\MST(S))$. What remains is to show that both $R$ and $B$ are non-crossing.

The $\MST^3(S)$-edge $v_3v_0$ from $E_{PB}$ is also the only edge that could cause a crossing, see Lemma~\ref{lem:MST2edges} and Theorem~\ref{thm:construction}.
Hence, $R$ is a plane spanning tree.
If $v_3v_0$ is not crossed by any other edge of $E_{B}$ then also $B$ is a plane spanning tree and we are done.
Otherwise, note first that by Lemma~\ref{lem_not_in_triangle}, the triangles $v_3v_2v_1$ and $v_2v_1v_0$ cannot contain any points of $S$. 
Then observe that for the cases in \fig{fig:combinedRB1}~(b),~(c),~(e) and~(f) any edge crossing $v_3v_0$ that does not have $v_0,v_1,v_2$ or $v_3$ as an endpoint must cross an \MST-edge between $v_0,v_1,v_2$ and $v_3$. 
This implies that any $\MST^2$-edge that crosses $v_3v_0$ must have $v_0,v_1,v_2$ or $v_3$ as its witness by Lemma~\ref{lem:MST2edges}.
By construction, $v_0,v_1,v_2$ and $v_3$ are not a witness to any blue edge between vertices in $S \setminus \{v_0,v_1,v_2,v_3\}$.
The edges in $E_{PB}$ are incident to $v_0$ or $v_3$ so they also cannot cross $v_3v_0$.

For cases from \fig{fig:combinedRB1}~(a) and~(d), as well as all cases from \fig{fig:combinedRB1nonconv}, 
observe first that $v_0v_1$, $v_1v_2$, and $v_2v_3$ cannot be crossed, again due to Lemma~\ref{lem:MST2edges} and the fact that $v_0,v_1,v_2$ and $v_3$ cannot be a witness to any long edge in the grey subtrees by construction. 
So any edge that crosses $v_3v_0$ must have an endpoint in the interior of the convex hull of $P$ or connect directly to $v_1$ or $v_2$.
The latter however cannot happen: The only points connecting with blue edges to $v_1$ or $v_2$ are direct neighbors of these vertices, which reside in the large-angled wedge $v_1v_2v_3$ or $v_0v_1v_2$, respectively.
(Here, the \emph{large-angled} wedge $v_1v_2v_3$ is the wedge spanned by the a ray from $v_2$ to $v_1$ and a ray from $v_2$ to $v_3$ so that the wedge has an opening angle larger than~$\pi$.
The large-angled wedge $v_0v_1v_2$ is defined analogously.)
Hence, if $v_0v_3$ is crossed by some blue edge, there must be a nonempty set $X\subset S\setminus P$ that resides in the interior of the convex hull of $P$. 
In the cases depicted in \fig{fig:combinedRB1}~(a) and~(d), $X$ lies in the triangle $\Delta$ spanned by $v_3$, $v_0$, and the intersection of $v_3v_1$ and $v_2v_0$.
In the cases depicted in \fig{fig:combinedRB1nonconv}, the $X$ lies in the triangle $\Delta$ spanned by $v_3$, $v_0$, and the vertex of $P$ in the interior of the convex hull of $P$. 
Further, removing the edge $v_0v_3$ from $B$ splits $B$ into two connected components that are each plane trees, where $v_3$ is in one and $v_0$ is in the other component.
Now consider the convex hull of $X \cup \{v_0, v_3\}$, and the path along the boundary of that convex hull between $v_0$ and $v_3$ that contains at least one vertex of $X$. 
This path contains exactly one edge $e$ that connects the two components of $B$. 
Due to the construction of $B$ and $R$, $e$ can neither be part of $R$ (as the two endpoints of $e$ must reside in two different subtrees of $v_0,v_1$ or $v_2$) nor cross any edge of $B$ (as the only possibly intersected segment of the convex hull boundary of $X$ was $v_0v_3$). 
Further, the length of $e$ must be less than $3\BE(\MST(P))$, as $e$ lies inside the triangle $\Delta$, and as all sides of $\Delta$ are bounded from above by $3\BE(\MST(P))$. 
Hence, as $v_3v_0$ was the only edge that could be crossed by others, the replacement of $v_3v_0$ by $e$ in $B$ results in two edge-disjoint plane spanning trees $R$ and $B$ with maximum edge length less than $3\BE(\MST(P))$.

\bigskip
As for Case~\ref{thm:allpointed-item:2}~(b), $S$ consists only of four vertices and hence \fig{fig:colP}~(d) already shows all of the two trees $R$ and $B$, it remains to consider the subtree for Case~\ref{thm:allpointed-item:2}~(a).

\begin{figure}[htb]
  \centering
  \includegraphics[page=5]{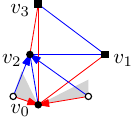}
  \caption{The plane spanning trees $R$ and $B$ for Case~\ref{thm:allpointed-item:2}.}
  \label{fig:combinedRB2}
\end{figure}

\textbf{$\bm{T_0}$:} Consider the subtree $T_0=G(S_0,E_0)$ of $\MST(S)$, with $S_0=S_0'\cup\{v_2\}$, $E_0=E_0'\cup\{v_2v_0\}$.
We root $T_0$ at $r = v_2$ (observe, $v_2$ is a leaf in $T_0$ with unique child $s=v_0$) and apply Construction~\ref{con:dir-red-blue-doubletrees} to obtain edge sets $E'_{0R}$ and $E'_{0B}$, since we will add connectivity between $r=v_2$ and $s=v_0$ using $E_{PR}$ and $E_{PB}$ we remove the edge $rs$ to obtain $R_0=G(S_0,E_{0R})$ and $B_0=G(S_0,E_{0B})$, with $E_{0R} = E'_{0R} \setminus \{rs\}$ and $E_{0B} = E'_{0B} \setminus \{rs\}$.
We use the inverted coloring as defined in Lemma~\ref{lem:red-blue-tree-variants} for the two trees of $T_0$, implying that
the edges connecting to $v_2$ and crossing red edges of $E_{PR}$ are all blue.
Hence the graphs $R=G(S,E_{PR}\cup E_{0R})$ and $B=G(S,E_{PB}\cup E_{0B})$ are plane spanning trees, $E_R\cap E_B=\emptyset$, and $\max\{\BE(R), \BE(B)\}\leq 2\BE(\MST(S))$.

This concludes the proof.
\end{proof}

\begin{corollary}
For any set $S$ of $n\geq 4$ points in the plane, there are two plane spanning trees $R=G(S,E_R)$ and $B=G(S,E_B)$ such that $E_R\cap E_B=\emptyset$ and  $\max\{\BE(R), \BE(B)\}\leq 3\BE(\MST(S))$.
\end{corollary}

We now show that the above construction is worst-case optimal. 

\begin{theorem}\label{Lower-Bounds}
For any $n>3$ and $k > 1$ there is a set $S$ of $n$ points such that for any $k$ disjoint spanning trees, at least one has a bottleneck edge larger than $(k+1) \BE(\MST(S))$.
\end{theorem}
\begin{proof}
A counterexample simply consists of $n$ points equally distributed on a line segment. The points can be slightly perturbed to obtain general position (similar to \fig{fig:construc1}).
In this problem instance there are $kn-(k(k+1)/2)$ edges whose distance is strictly less than $(k+1 )\BE(\MST(S)) = k+1$.
However, we need $kn - k$ edges for $k$ disjoint trees and thus it is impossible to construct that many trees with sufficiently short edges.
\end{proof}

\section{Distributed Approach}\label{sec_dist}
The previous construction relies heavily on the minimum spanning tree
of $S$. It is well known that this tree cannot be constructed locally,
thus we are implicitly assuming that the network is constructed by a
single processor that knows the location of all other
vertices.  In ad-hoc
networks, it is often desirable that each vertex can compute its
adjacencies using only local information, i.e., using only information about vertices at most a certain maximum distance away.

In the following, we provide an alternative construction. Although, for any fixed $k$, the length of the edges is increased by a constant factor of $12 \sqrt{2} k$ (see Theorem~\ref{thm:main2} below for details), it has the benefit that it can be constructed locally and that it can be extended to compute $k$ layers, for $k \leq n/12$. The only global property that is needed is a value~$\beta$ that should be at least $\BE(\MST(S))$. We also note that these plane disjoint graphs are not necessarily trees, as large cycles cannot be detected locally. 

Before we describe our approach, we report the result of
Biniaz and Garc\'ia~\cite{biniaz_garcia}
that every point set of at least $3k$
points contains $k$ layers. Since the details of this construction are
important for our construction, we add a proof sketch.

\begin{theorem}[\cite{biniaz_garcia}]
\label{theo_layers}
  Every finite point set that consists of at least $3k$ points contains $k$ layers. 
\end{theorem}
\begin{proof}
First, recall that for every set of $n$ points, there is a
\emph{center point} $c$ such that every line through $c$ splits the point set into two parts that each
contain at least $\lfloor n/3 \rfloor$ points, see e.g.\ Chapter 1 in~\cite{matousek} (note that $c$ need not be one of the initial $n$ points).  For ease of explanation, we assume
that every line through $c$ contains at most one point.
 Number the points $v_0,v_1,\ldots,v_{n-1}$ in clockwise circular order
around $c$. 
We split the plane into three angular regions by the three rays
originating from $c$ and passing
through $v_0$, $v_{\floor{\frac{n}{3}}}$, and $v_{\floor{\frac{2n}{3}}}$;
see \fig{fig:ConnectingInternally}. 
Since every line through the center contains
at least $n/3$ points on each side, the smaller angle of the two rays defining a region is at most $\pi$ and thus the three angular regions are convex.
We declare $v_0$ to be the representative of the angular region
between the rays
 through $v_0$ and $v_{\lfloor
  \frac{n}{3} \rfloor}$ and connect the vertices $v_1, ..., v_{\lfloor
  \frac{n}{3} \rfloor}$ in this region to $v_0$.  Similarly, we assign $v_{\lfloor
  \frac{n}{3} \rfloor}$ to be the representative of angle between the
rays center through $v_{\lfloor \frac{n}{3} \rfloor}$ and
$v_{\lfloor \frac{2n}{3} \rfloor}$ and connect vertices $v_{\lfloor
  \frac{n}{3} \rfloor + 1}, ..., v_{\lfloor \frac{2n}{3} \rfloor}$ to
$v_{\lfloor \frac{n}{3} \rfloor}$. Finally, we
 connect vertices $v_{\lfloor \frac{2n}{3}
  \rfloor + 1}, ..., v_{n-1}$ to $v_{\lfloor \frac{2n}{3} \rfloor}$.
This results in a non-crossing spanning tree.

For the second tree, we rotate the construction and we use $v_1$, $v_{\lfloor \frac{n}{3}
  \rfloor + 1}$, and $v_{\lfloor \frac{2n}{3} \rfloor + 1}$ to define the three regions,
and so on.
\end{proof}

\begin{figure}[ht]
  \centering
  \includegraphics{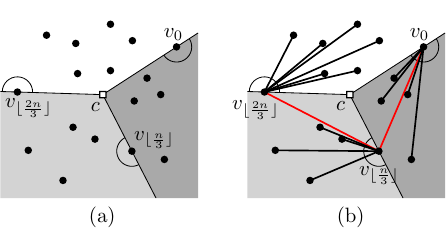}
  \caption{Extracting one layer: (a) The three sectors defined by $v_0$, $v_{\lfloor \frac{n}{3} \rfloor}$, and $v_{\lfloor \frac{2n}{3} \rfloor}$. (b) Connecting the points to the representative of their sector. The red edges connect the representatives.}
  \label{fig:ConnectingInternally}
\end{figure}

While this construction provides a simple method of constructing the $k$ layers, it does not give any guarantee on the length of the longest edge in this construction. To give such a guarantee, we combine it with a bucketing approach: we partition the point set using a grid (whose size will depend on $k$ and $\beta$), solve the problem in each box with sufficiently many points independently, and then combine the subproblems to obtain a global solution (see \fig{fig:DistributedApproach}).

\begin{figure}[ht]
  \centering
  \includegraphics{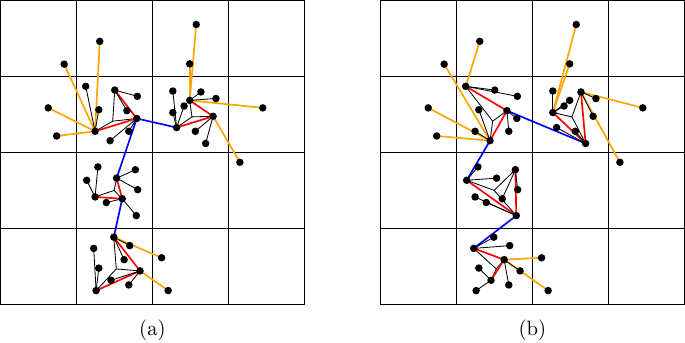}
  \caption{The distributed approach: a grid is placed over the point set and different representatives construct different graphs ((a) and (b)). The red and black edges form the tree in each dense cell, blue edges connect the dense cells, and orange edges connect the vertices in sparse cells.}
  \label{fig:DistributedApproach}
\end{figure}

We place a grid with cells of height and width $6k\beta$ and classify the points according to which grid cell contains them (if a point lies exactly on a separating line, we assign it an arbitrary adjacent cell).
We say that a grid cell is a \emph{dense box} if it contains at least $3k$ points of $S$.
Similarly, it is a \emph{sparse box} if it contains points of $S$ but is not dense.
Two boxes are adjacent if they share some boundary or vertex. Hence, each box has 8 neighbors. This is also referred to as the 8-neighbor topology. 
We observe that dense and sparse boxes satisfy the following properties.

\begin{lemma}\label{lem_nonadjacent}
Given two non-adjacent boxes $B$ and $B'$, the points in $B$ and $B'$ cannot be connected by edges of length at most $\beta$ using only points from sparse boxes.
\end{lemma}
\begin{proof}
Suppose the contrary and let $B$ and $B'$ be two boxes such that there is a path that uses edges of length at most $\beta$ between a point in $B$ and a point in $B'$ visiting only points in sparse boxes.
This path crosses the sides of a certain number of boxes in a given order; let $\sigma$ be the sequence of these sides, after repeatedly removing adjacent duplicates.
Observe first that horizontal and vertical sides alternate in~$\sigma$, as otherwise the path would traverse the cell width of $6k\beta$ using at most $3k-1$ points connected by edges of length at most $\beta$.
Since $B$ and $B'$ are non-adjacent, w.l.o.g., there is a vertical side~$s$ that has two adjacent horizontal sides in~$\sigma$ with different $y$-coordinates.
Hence, between the two horizontal sides, the corresponding part of the path has length at least $6k\beta$, and may use only the points in the two boxes adjacent to~$s$.
But since any sparse box contains at most $3 k - 1$ points and the distance between two consecutive points along the path is at most $\beta$, that part of the path can have length at most $(6k-1) \beta$, a contradiction.
\end{proof}

\begin{corollary}\label{cor_8neighbor}
Dense boxes are connected by the 8-neighbor topology. 
\end{corollary}

\begin{lemma}
Any finite set~$S$ of at least $4 \cdot (3k -1) + 1$ points with $\beta \geq \BE(\MST(S))$ contains at least one dense box. 
\end{lemma}
\begin{proof}
Assume $S$ and $\beta$ induce only sparse boxes.
This implies that the points are distributed over at least five boxes, and thus, there is a pair of boxes that is non-adjacent.
Using Lemma~\ref{lem_nonadjacent}, this means that these boxes cannot be connected using edges of length at most $\BE(\MST(S))$, a contradiction.
\end{proof}

\begin{lemma}\label{lem_sparsenotfar}
In any finite set~$S$ of at least $4 \cdot (3k -1) + 1$ points with $\beta \geq \BE(\MST(S))$, all sparse boxes are adjacent to a dense box.
\end{lemma}
\begin{proof}
This follows from Lemma~\ref{lem_nonadjacent}, since any sparse box that is not adjacent to a dense box cannot be connected to any dense box using edges of length at most $\beta \geq \BE(\MST(S))$. 
\end{proof}

Next, we assign all points to dense boxes.
In order to do this, let $c_B$ be the center of a dense box $B$.
Note that $c_B$ is not necessarily the center point of the points in this box.
We consider the Voronoi diagram of the centers of all dense boxes and assign a point $p$ to $B$ if $p$ lies in the Voronoi cell of $c_B$.
Let $S_B$ be the set of points of $S$ that are associated with a dense box $B$.
We note that each dense box $B$ gets assigned at least all points in its own box, since in the case of adjacent dense boxes, the boundary of the Voronoi cell coincides with the shared boundary of these boxes (see \fig{fig:VoronoiCells}).

\begin{figure}[ht]
  \centering
  \includegraphics{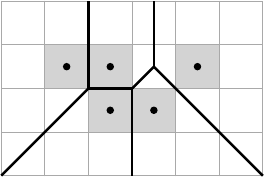}
  \caption{The Voronoi cells of the centers of the dense boxes.}
  \label{fig:VoronoiCells}
\end{figure}

Furthermore, we can compute the points assigned to each box locally.
By Lemma~\ref{lem_sparsenotfar} all sparse boxes are adjacent to a dense box, and hence for any point~$p$ in a sparse box~$B$ its distance to its nearest center is at most $3\ell / \sqrt{2}$, where $\ell = 6k\beta$.
It follows that only the centers of cells of neighbors and neighbors of neighbors need to be considered.

\begin{lemma}\label{lem_convexdisjoint}
  For any two dense boxes $B$ and $B'$, we have that the convex hulls of $S_B$ and $S_{B'}$ are disjoint. 
\end{lemma}
\begin{proof}
We observe that the convex hull of $S_B$ is contained in the Voronoi cell of $c_B$.
Hence, since the Voronoi cells of different dense boxes are disjoint, the convex hulls of the points assigned to them are also disjoint.
\end{proof}

For each dense box $B$, we apply Theorem~\ref{theo_layers} on the points inside the dense box to compute $k$ disjoint layers of $S_B$. Next, we connect all sparse points in $S_B$ to the representative of the sector that contains them in each layer. Since all points in the same sector connect to the same representative and the sectors of the same layer do not overlap, we obtain a plane graph for each layer within the convex hull of each $S_B$. 

Hence, we obtain $k$ pairwise disjoint layers such that in each layer the points associated to each dense box are connected. Moreover, since the created edges stay within the convex hull of each subproblem and by Lemma~\ref{lem_convexdisjoint} those hulls are disjoint, each layer is plane.
Thus, to assure that each layer is connected, we must connect the construction between dense boxes.

We connect adjacent dense boxes in a tree-like manner using the following rules:

\begin{itemize}
  \item Connect every dense box to any dense box below it. 
  \item Always connect every dense box to any dense box to the left of it. 
  \item If neither the box below nor the one to the left of it is dense, connect the box to the dense box diagonally below and to the left of it. 
  \item If neither the box above nor the one to the left of it is dense, connect the box to the dense box diagonally above and to the left of it. 
\end{itemize}

To connect two dense boxes, we find and connect two representatives~$p$ and~$q$ (one from each dense box) such that~$p$ lies in the sector of~$q$ and~$q$ lies in the sector of~$p$; see \fig{fig:connecting_joined}~(a).

\begin{figure}[ht]
  \centering
  \includegraphics{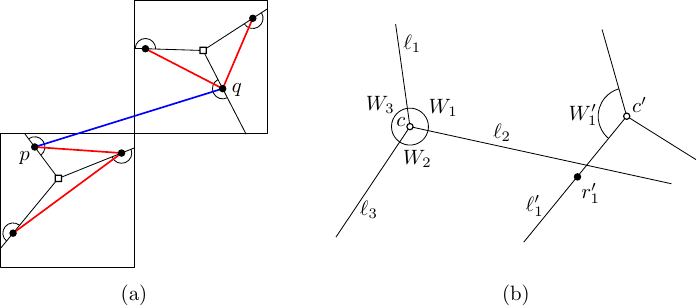}
  \caption{Connecting two dense boxes by means of $p$ and $q$. The half-circles in (a) indicate which sector each representative covers. The red edges connect the dense boxes internally and the blue edge connects the two dense cells.
  (b) illustrates the sectors involved in connecting two neighboring dense boxes.}
  \label{fig:connecting_joined}
\end{figure}

\begin{lemma}\label{lem:connect_boxes}
For any layer and any two adjacent dense boxes $B$ and $B'$, there are two representatives $p$ and $q$ in $B$ and $B'$, respectively, such that $p$ lies in the sector of $q$ and $q$ lies in the sector of~$p$.
\end{lemma}
\begin{proof}
Consider two boxes $B$ and $B'$ with center points (of their respective point sets)~$c$ and~$c'$.
Now let $W_1$ and $W_1'$ with representatives $r_1$ and $r_1'$ denote the sectors containing $c'$ and $c$, respectively;
see \fig{fig:connecting_joined}.
The other sectors $W_2$ and $W_3$ of $B$ with representatives $r_2$ and $r_3$ are ordered clockwise. We use $\ell_i$ to denote the ray from $c$ containing $r_i$. If $r_1\in W_1'$ and $r_1' \in W_1$ we are done. So assume that $r_1' \not \in W_1$, the case when $r_1 \not\in W'_1$ (or when both $r_1 \not\in W'_1$ and $r_1' \not \in W_1$) is symmetric. It follows that $r_1'$ is in sector $W_2$ if the line segment $c'r_1'$ intersects $\ell_2$ or sector $W_3$ if the segment intersects $\ell_2$ and $\ell_3$. Assume that $r_1'$ is in sector $W_2$ (again the argument is symmetric when $r_1'$ is in sector $W_3$). Now $r_2$ can be positioned on $\ell_2$ between $c$ and the intersection point with $c'r_1'$ or behind this intersection point when viewed from $c$. In the former case $r_1'$ is in $W_2$ and $r_2$ is in $W_1'$ and we are done. In the latter case the segments $cr_2$ and $c'r_1'$ cross. Since $c, r_2 \in B$ and $c', r_1' \in B'$ this crossing would imply that $B$ and $B'$ are not disjoint, a contradiction.
\end{proof}

Now that we have completed the description of the construction, we show that each layer of the resulting graph is plane and connected, and that the length of the longest edge is bounded. 

\begin{lemma}
  Each layer is plane. 
\end{lemma}
\begin{proof}
Since dense boxes are internally plane and the addition of edges to the sparse points do not violate planarity, it suffices to show that the edges between dense boxes cannot cross any previously inserted edges and that these edges cannot intersect other edges used to connect dense boxes.
  
We first show that the edge used to connect boxes $B$ and $B'$ is contained in the union of the Voronoi cells of these two boxes.
If $B$ and $B'$ are horizontally or vertically adjacent, the connecting edge stays in the union of the two dense boxes, which is contained in their Voronoi cells.
If $B$ and $B'$ are diagonally adjacent, we connect them only if their shared horizontal and vertical neighbors are not dense.
This implies that at least the two triangles defined by the sides of $B$ and $B'$ that are adjacent to their contact point are part of the union of the Voronoi cells of these boxes.
Hence, the edge used to connect $B$ and $B'$ cannot intersect the Voronoi cell of any other dense box.
Since all points of a dense box in a sector connect to the same representative and these edges lie entirely inside the sector, the edge connecting two adjacent boxes can intersect only at one of the two representatives, but does not cross them.
Therefore, an edge connecting two adjacent dense boxes by connecting the corresponding representatives cannot cross any previously inserted edge.
  
Next, we show that edges connecting two pairs of dense boxes cannot cross.
Since any edge connecting two dense boxes stays within the union of the Voronoi cells of $B$ and $B'$, the only way for two edges to intersect is if they connect to the same box $B$ and intersect in the Voronoi cell of $B$.
If the connecting edges lie in the same sector of $B$, they connect to the same representative and thus they cannot cross.
If they lie in different sectors of $B$, the edges lie entirely inside their respective sectors.
Since these sectors are disjoint, this implies that the edges cannot intersect. 
\end{proof}

\begin{lemma}
  Each layer is connected. 
\end{lemma}
\begin{proof}
  Since the sectors of the representatives of the dense boxes cover the plane, each point in a sparse box is connected to a representative of the dense box it is assigned to. Hence, showing that the dense boxes are connected completes the proof. 
  
By Corollary~\ref{cor_8neighbor}, the dense boxes are connected using the 8-neighbor topology.
This implies that there is a path between any pair of dense boxes where every step is one to a horizontally, vertically, or diagonally adjacent box.
Since we always connect horizontally or vertically adjacent boxes and we connect diagonally adjacent boxes when they share no horizontal and vertical dense neighbor, the layer is connected after adding edges as described in the proof of Lemma~\ref{lem:connect_boxes}.
\end{proof}

\begin{lemma}
  The distance between a representative in a dense box $B$ and any point connecting to it is at most $12 \sqrt{2} k \beta$.
\end{lemma}
\begin{proof}
  Since the representatives of $B$ are connected only to points from dense and sparse boxes adjacent $B$, the distance between a representative and a point connected to it is at most the length of the diagonal of the $2 \times 2$ grid cells with $B$ as one of its boxes. Since a box has width $6 k \beta$, this diagonal has length $2 \sqrt{2} \cdot 6 k \beta = 12 \sqrt{2} k \beta$.
\end{proof}

\begin{theorem}\label{thm:main2}
  For all finite point sets with at least $4 (3k - 1) + 1$ points, we can extract $k$ plane layers with the longest edge having length at most $12 \sqrt{2} k \BE(\MST(S))$.
\end{theorem}

\section{Conclusions}
We presented two algorithms for constructing $k$ edge-disjoint non-crossing plane spanning graphs on a given point set such that the length of the bottleneck edge is minimized. The first algorithm uses global properties in order to keep all edges as small as possible. We also give matching worst-case lower bounds, making the algorithm tight. The main drawback is that this method can only be used to construct two layers, and it is unlikely that a similar approach can work for more.  Our second algorithm works for a large number of layers (up to $n/12$ layers). It uses only local information, thus it can be executed in a distributed manner. The drawback of this approach is that the length of the bottleneck edge grows considerably: for two layers, the $24 \sqrt{2} \BE(\MST(S))$ implied by this method is far larger than the $3 \BE(\MST(S))$ of the first approach.


So far, there is no centralized method to construct
more than two trees. Finding such a method and comparing
it to the distributed method presented here is
an interesting direction of future research. Another direction would
be to lower the length of the longest edge in the distributed
construction, though from a purely worst-case theoretical point of
view this is likely to require a different approach from the one used
in this paper.

\subsection*{Acknowledgments} This research was initiated during the 10th European Research Week on Geometric Graphs (GGWeek 2013), Illgau, Switzerland. We would like to thank all participants for fruitful discussions.
O.A., A.P., and B.V.\ were partially supported by the ESF EUROCORES programme EuroGIGA - ComPoSe, Austrian Science Fund (FWF): I~648-N18. T.H.\ was supported by the Austrian Science Fund (FWF): P23629-N18 `Combinatorial Problems on Geometric Graphs'. M.K.\ was supported in part by the ELC project (MEXT KAKENHI No.17K12635) and NSF award CCF-1423615..
A.P.\ is supported by an Erwin Schr\"odinger fellowship, Austrian Science Fund (FWF): J-3847-N35.
A.v.R.\ and M.R.\ were supported by JST ERATO Grant Number JPMJER1201, Japan. 
\bibliographystyle{splncs03} %
\bibliography{packing}

\begin{thebibliography}{10}
\providecommand{\url}[1]{\texttt{#1}}
\providecommand{\urlprefix}{URL }

\bibitem{ahkk-ppst-14}
Aichholzer, O., Hackl, T., Korman, M., van Kreveld, M.J., L{\"{o}}ffler, M.,
  Pilz, A., Speckmann, B., Welzl, E.: Packing plane spanning trees and paths in
  complete geometric graphs. Inf. Process. Lett.  124,  35--41 (2017)

\bibitem{DBLP:conf/isaac/AichholzerHKPRR16}
Aichholzer, O., Hackl, T., Korman, M., Pilz, A., Rote, G., van Renssen, A.,
  Roeloffzen, M., Vogtenhuber, B.: Packing short plane spanning trees in
  complete geometric graphs. In: {ISAAC}. LIPIcs, vol.~64, pp. 9:1--9:12.
  Schloss Dagstuhl - Leibniz-Zentrum fuer Informatik (2016)

\bibitem{biniaz_garcia}
Biniaz, A., Garc{\'{\i}}a, A.: Packing plane spanning trees into a point set.
  In: Proc. 30th Canadian Conference on Computational Geometry, ({CCCG}). pp.
  49--53 (2018)

\bibitem{dt-gdnpc-97}
Dor, D., Tarsi, M.: Graph decomposition is {NP}-complete: {A} complete proof of
  {H}olyer's conjecture. {SIAM} J. Comput.  26(4),  1166--1187 (1997)

\bibitem{fwz-iar-05}
Fussen, M., Wattenhofer, R., Zollinger, A.: Interference arises at the
  receiver. In: Proc. International Conference on Wireless Networks,
  Communications, and Mobile Computing (WIRELESSCOM). pp. 427--432 (2005)

\bibitem{k-mianbcr-12}
Korman, M.: Minimizing interference in ad-hoc networks with bounded
  communication radius. Inf. Process. Lett.  112(19),  748--752 (2012)

\bibitem{ksu-crgn-99}
Kranakis, E., Singh, H., Urrutia, J.: Compass routing on geometric networks.
  In: Proc. 11th Canadian Conference on Computational Geometry (CCCG). pp.
  51--54 (1999)

\bibitem{matousek}
Matou{\v{s}}ek, J.: Lectures on Discrete Geometry. Graduate Texts in
  Mathematics, Springer (2002)

\bibitem{pt-mdsm-05}
Priesler, M., Tarsi, M.: Multigraph decomposition into stars and into
  multistars. Discrete Mathematics  296(2-3),  235--244 (2005)

\bibitem{tarsi-dcmsp-83}
Tarsi, M.: Decomposition of a complete multigraph into simple paths:
  Nonbalanced handcuffed designs. J. Comb. Theory, Ser. {A}  34(1),  60--70
  (1983)

\end{thebibliography}

\end{document}